\documentclass[11p,reqno]{amsart}

\usepackage[T1]{fontenc}
\usepackage{graphicx}
\usepackage{amssymb,amscd,amsthm,amsmath,mathrsfs,bm,braket,marginnote}
\usepackage{enumerate}
\usepackage{appendix}
\usepackage[colorlinks=true, pdfstartview=FitV, linkcolor=blue, citecolor=blue, urlcolor=blue]{hyperref}

\usepackage{pgf}
\usepackage{pgfplots}
\usepackage{tikz}
\usetikzlibrary{arrows,calc}
\usepackage{verbatim}
\usetikzlibrary{decorations.pathreplacing,decorations.pathmorphing}
\usepackage[numbers,sort&compress]{natbib}

\usepackage{array,arydshln}

\numberwithin{equation}{section}
\newtheorem{theorem}{Theorem}[section]
\newtheorem{lemma}[theorem]{Lemma}

\newtheorem{remark}[theorem]{Remark}

\newtheorem{proposition}[theorem]{Proposition}

 \reversemarginpar
 \newcommand{\pf}{\text{pf}}
 \newcommand{\Pf}{\text{Pf}}
 \newcommand{\cF}{{\mathcal F}}
 \newcommand{\bC}{{\mathbb C}}
 \newcommand{\rz}{|0\rangle}
\newcommand{\lz}{\langle0|}
\newcommand{\cM}{\mathcal{M}}
\begin{document}
\title[BKP hierarchy and Pfaffian point process]{BKP hierarchy and Pfaffian point process}

\thanks{$^\ddag$Department of Mathematics, China University of Mining and Technology, Beijing 100083, China}

\thanks{$^\dag$ Department of Mathematics and Statistics, University of Melbourne, Victoria 3010, Australia}

\thanks{$^\ast$Corresponding author: Shi-Hao Li(shihao.li@unimelb.edu.au)}

\author[Zhi-Lan Wang]{Zhi-Lan Wang $^{\ddag}$}
\author[Shi-Hao Li]{Shi-Hao Li $^{\dag,\ast}$}

\thanks{Email addresses: zlwang@cumtb.edu.cn; shihao.li@unimelb.edu.au}

\date{}

\dedicatory{}


	\begin{abstract}
Inspired by Okounkov's work [\emph{Selecta Mathematica}, 7(1):57--81, 2001] which relates KP hierarchy to determinant point process, we establish a relationship between BKP hierarchy and Pfaffian point process. We prove that the correlation function of the shifted Schur measures on strict partitions can be expressed as a Pfaffian of skew symmetric  matrix kernel, whose elememts are certain vacuum expectations of neutral fermions. We further show that the  matrix integrals solution of BKP hierarchy can also induce a certain Pfaffian point process.
	\end{abstract}
	
	\maketitle

\section{Introduction}
There is a connection revealed by Okounkov \cite{Okounkov2001Infinite} between random partitions and KP hierarchy of type $A_\infty$. He introduced the Schur measure of a partition $\lambda$ which is proportional to $s_{\lambda}(x)s_{\lambda}(y)$, where $s_{\lambda}$ is the Schur function. Then the correlation function can be realized as a determinant point process via the Fock space formalism, which satisfies the KP (or 2d-Toda lattice) hierarchy. In Kyoto school's picture, there are different ways to describe the tau function of KP hierarchy of type $A_\infty$ \cite{alexandrove13,Date1981Vertex,JM}. On one hand, the tau function can be viewed as an element in the Bosonic Fock space (the space of symmetric functions), and thus it can be expressed in terms of Schur funcions. On the other hand, the tau funcion can also be viewed as a vector in the Fermionic Fock space (the infinite wedge space). The Boson-Fermion correspondence is an explicit isomorphism between these two spaces, via which a relationship between Schur functions and vacuum expectations of fermions is obtained, and this is the reason why the Boson-Fermion correspondence of type $A_\infty$ plays an important role in Okounkov's approach. The infinite wedge space also appears in \cite{borodin05,okounkov2003correlation} on Schur process and in \cite{vanmoerbeke2011Random}  on a Toeplitz-type determinanat point process.

Moreover, besides KP hierarchy of type $A_\infty$, there are KP hierarchies of other types. One of them is of type $B_\infty$, which we call BKP hierarchy for short. BKP hierarchy was introduced and studied in details by the Kyoto school \cite{date82,JM}. Many phenomenons in BKP hierarchy are parallel to the cases of type $A_\infty$. For example, In \cite{Okounkov00} Okounkov found that the generating functions for the Hurwitz numbers of certain ramified coverings of $\bC\mathbb{P}^1$ are the 2d-Toda tau functions, and parallelly Natanzon and Orlov \cite{natanzon15} showed that the generating functions for the weighted Hurwitz numbers of certain branched coveings of $\mathbb{R}\mathbb{P}^2$ are BKP tau functions. Another analogy between type $A_\infty$ and $B_\infty$ is also important: in the case of type $B_\infty$, the corresponding tau function can be described in the language of neutral fermions. Therefore, many results derived by using free fermions in the $A_\infty$ case can be generalized to the $B_\infty$ case, if one instead uses neutral fermions. For instance, In \cite{foda07} neutral fermions are used to obtain volume-weighted plane partitions, which is analogous to the fact that free fermions can be used to obtain plane partitions \cite{okounkov2003correlation}. Moreover, as the BKP tau functions can be described respectively in the Bosonic picture and the Fermionic picture \cite{you89}, the Boson-Fermion correspondence of type $B_\infty$ allows us to relate projective Schur functions to vacuum expectations of neutral fermions.

Inspired by these facts, we generalize Okounkov's results to BKP hierarchy in this paper. We consider the shifted Schur measure $\cM$ on strict paritions \cite{tracy2004limit}, whose weight of $\lambda$ is propotional to $P_\lambda(x)Q_\lambda(y)$, where $P_\lambda (x)$ and $Q_\lambda (y)$ are projective Schur functions. For a  finite set $A \in\mathbb{Z}_+$, the correlation function is defined as the probability that the set ${\sigma}(\lambda)=\{\lambda_i\}$ containing $A$. In Theorem \ref{thm:Paf}, we prove that
\begin{align}\label{ppp}
\rho(A)=\Pf(K(a_i,a_j))_{a_{ i},a_{ j}\in\pm A},
\end{align}
and thus we can relate the BKP hierarchy to a Pfaffian point process.

The same result has been obtained by Matsumoto \cite{matsumoto2005correlation}, who calculated the correlation function using operators on the exterior algebra. Later, Vuleti{\'c} \cite{vuletic2007shifted} generalized Matsumoto's result to the shifted Schur process and related it to a Pfaffian point process. Both these two articles are generalizations of Okounkov's results \cite{Okounkov2001Infinite,okounkov2003correlation}, and they mainly focused on the measures of strict partitions. However, in our article, we start with a different point of view. We use the neutral fermions to do calculations and pay more attention to its connection with the theory of integrable hierarchy. Moreover, since the tau function of BKP hierarchy plays a key role in Pfaffian point process, we turn to a specific $\tau$-function------the matrix integrals solution to BKP hierarchy and find its connection with some certain Pfaffian point process. Interestingly, an explicit Pfaffian point process with skew symmetric matrix kernel is induced by the matrix integrals solution of BKP hierarchy, which generalizes the Pfaffian point process in \cite{rains00} and unify a certain Pfaffian point process with arbitrary points rather than the points of even number.

The rest of this paper is arranged as follows. In Section \ref{sec:Abr}, we review some basic facts on neutral fermions and the Fock space for BKP hierarchy. In particular, we introduce the Boson-Fermion correspondence of type $B_\infty$. In Section \ref{sec:Neu}, projective Schur functions are introduced in terms of neutral fermions and we use them to construct a measure on strict partitions, thus proving that the correlation functions can be realized as a Pfaffian point process. In Section \ref{sec:Mat}, a certain Pfaffian point process is constructed from the matrix integrals solution of BKP hierarchy.

\section{A brief introduction to neutral fermions}\label{sec:Abr}
In \cite{JM}, neutral fermions $\{\phi_n, n\in \mathbb{Z}\}$ were introduced to construct the spin representation of $B_\infty$ which derive the BKP hierarchy finally. In this section, we give a brief introduction to the method proposed by Jimbo and Miwa.

Introducing the neutral free fermions $\{\phi_n,n\in\mathbb{Z}\}$ by the relations
\begin{align*}
\phi_n=\frac{\psi_n+(-1)^n\psi^*_{-n}}{\sqrt{2}},
\end{align*}
where $\{\psi_i,\psi_i^*, \,i\in\mathbb{Z}\}$ are standard charged free fermions satisfying the anti-involution relations $[\psi_i,\psi^*_j]_+=\delta_{i,j}$ and $[\psi_i,\psi_j]_+=[\psi^*_i,\psi^*_j]_+=0$, it turns out
\begin{align*}
[\phi_m,\phi_n]_+=(-1)^m\delta_{m,-n}.
\end{align*}

Let $W_B$ be the complex space spanned by $\{\phi_n,n\in\mathbb{Z}\}$ and denote $Cl(W_B)$ as the Clifford algebra generated by $W_B$, then the right Fock space for neutral fermions can be defined by
$$\cF_B=Cl(W_B)/Cl(W_B)(\sum\limits_{n<0}\bC\phi_n),$$
and we denote $|0\rangle$ as the residue class of 1 in $\cF_B$.
Similarly the left Fock space can be defined by
$$\cF_B^*=Cl(W_B)/(\sum\limits_{n>0}\bC\phi_n)Cl(W_B),$$
and $\langle0|$ is the residue class of 1 in $\cF_B^*$.

Clearly, from the above definition, we have 
\begin{equation*}
\phi_m|0\rangle=0,\;m<0 \quad\text{and}\quad
\langle0|\phi_m=0,\;m>0,
\end{equation*}
and thus \begin{align*}
\cF_B=\text{span}\{\phi_{n_1}\cdots\phi_{n_k}|0\rangle\},\quad \cF_B^*=\text{span}\{\langle0|\phi_{-n_k}\cdots\phi_{-n_1}\},\quad \text{with $n_1>\cdots>n_k\geq0$}.
\end{align*}

There is a nondegenerate bilinear pairing $\cF_B^*\times\cF_B\rightarrow \bC$, and we write the pairing of $\langle U|$ and $|V\rangle$ as $\langle U|V\rangle$. The vacuum expectation value of an operator $A$ is defined as $\langle 0|A|0\rangle$, and is denoted as $\langle A\rangle$.
We have
\begin{equation*}
\langle\phi_m\phi_n\rangle=\left\{\begin{array}{lcl}
(-1)^m\delta_{m,-n},&&n>0,\\ 
\frac{1}{2}\delta_{m,0},&&n=0,\\
0,&&n<0.
\end{array}
\right.
\end{equation*}
By using Wick's theorem, one has
\begin{equation}\label{eqn:VacExp}
\langle\phi(z_1)\cdots\phi(z_{2s})\rangle=\displaystyle\frac{1}{2^s}\prod_{j<j'}\displaystyle\frac{1-z_{j'}/z_j}{1+z_{j'}/z_j}.
\end{equation}

Moreover, a Hamiltonian in terms of neutral fermions can be defined as
\begin{align}\label{ham}
H_n=\frac{1}{2}\sum_{i\in\mathbb{Z}}(-1)^{i+1}\phi_i\phi_{-i-n}.
\end{align}
\begin{remark}
Here we would like to mention that $H_n$ has an original definition as 
\begin{align*}
H_n=\frac{1}{2}\sum_{i\in\mathbb{Z}}(-1)^{i+1}:\phi_i\phi_{-i-n}: \quad\text{if}\quad n\not=0
\end{align*}
where $:\cdot:$ is the normal order of fermions. This is because
\begin{align*}
H_n=\frac{1}{2}\sum_{i\in\mathbb{Z}}:\phi_i\phi_{-i-n}:=\frac{1}{2}\sum_{i\in\mathbb{Z}}(\phi_i\phi_{-i-n}-\langle\phi_i\phi_{-i-n}\rangle).
\end{align*}
We have known that $\langle\phi_i\phi_{-i-n}\rangle=0$ if $n\not=0$, therefore it follows (\ref{ham}).
\end{remark}
\begin{remark}
In the neutral fermions case, $H_n$ equals zero if $n$ is a nonzero even number. The reason lies in the fact that for an arbitrary integer $i=n$, there exists $i'=-n-m$, such that $(-1)^{i+1}\phi_i\phi_{-i-n}+(-1)^{i'+1}\phi_{i'}\phi_{-i'-n}=(-1)^{m+1}[\phi_m,\phi_{-n-m}]_+=0$ if $n\in2\mathbb{Z}\backslash \{0\}$, which means $H_n=0$ if $n$ is a nonzero even integer. In the followings, we assume $n$ is an odd number without extra statement.
\end{remark}

In the next, some properties of the Hamiltonian $H_n$ are demonstrated. Firstly, it is shown that $\{H_n,n\in 2\mathbb{Z}+1\}$ form a Heisenberg algebra.

\begin{proposition}
\begin{align}\label{comm}
[H_n,H_m]=\frac{n}{2}\delta_{n+m,0}.
\end{align}
\end{proposition}
\begin{proof}
Since
\begin{align*}
[\phi_i\phi_j,\phi_k\phi_l]=(-1)^j\delta_{j,-k}\phi_i\phi_l-(-1)^i\delta_{i,-k}\phi_j\phi_l+(-1)^j\delta_{j,-l}\phi_k\phi_i-(-1)^i\delta_{i,-l}\phi_k\phi_j,
\end{align*}
then it follows
\begin{align*}
[H_n,H_m]&=\frac{1}{4}\sum_{i,j\in\mathbb{Z}}(-1)^{i+j}[\phi_i\phi_{-i-n},\phi_j\phi_{-j-m}]\\
&=\frac{1}{4}\sum_{j\in\mathbb{Z}}[(-1)^{j-n}\phi_{j-n}\phi_{-j-m}-(-1)^j\phi_{j-n}\phi_{-j-m}+(-1)^{j-n}\phi_j\phi_{-j-m-n}-(-1)^j\phi_j\phi_{-j-m-n}],
\end{align*}
noticing that $H_n\not=0$ only in the cases of $n\in2\mathbb{Z}+1$, therefore,
\begin{align*}
[H_n,H_m]=\frac{1}{2}\sum_{j\in\mathbb{Z}}(-1)^j[-\phi_{j-n}\phi_{-j-m}-\phi_j\phi_{-j-m-n}]
\end{align*}
which equals $0$ if $n+m\not=0$ and $\frac{n}{2}$ if $n+m=0$.
\end{proof}
Moreover, if we introduce the current operators
\begin{align*}
H_+(t)=\sum_{k\geq1,odd}t_kH_k,\quad H_-(t)=\sum_{k\geq1,odd}t_{-k}H_{-k},
\end{align*}
then from a direct computation, one could obtain
\begin{align}\label{z}
[H_+(t),H_{-}(t)]=\sum_{k\geq1,odd}\frac{k}{2}t_kt_{-k}.
\end{align}
In what follows, we introduce some properties of the current operator $H_+(t)$ (respectively $H_-(t)$) which are helpful for us to derive integrable systems and Pfaffian point processes.
\begin{proposition}
For current operators $H_+(t)$ and $H_-(t)$, they satisfy
\begin{align*}
H_+(t)|0\rangle=0,\quad\langle 0|H_-(t)=0.
\end{align*}
\end{proposition} 
\begin{proof}
Since $H_n=\sum_{i\in\mathbb{Z}}(-1)^{i+1}\phi_i\phi_{-i-n}$ $(n>0)$, then if $i<0$, it is obvious that $\phi_i\phi_{-i-n}|0\rangle=0$. At the same time, when $i>0,-i-n<0$, $\phi_i\phi_{-i-n}|0\rangle=0$ could also be verified. Noticing that $i>-n$ and $i<0$ could run over the whole integer axis, therefore $H_n|0\rangle=0$ if $n>0$. Then one can conclude $H_+(t)|0\rangle=0$, and by the same manner, $\langle0|H_-(t)=0$ is established.
\end{proof}
\begin{remark}
It directly follows that $e^{H_+(t)}|0\rangle=|0\rangle$ and $\langle0|e^{H_-(t)}=\langle0|$.
\end{remark}

\begin{proposition}
Denote $\phi(z)=\sum_{i\in\mathbb{Z}}\phi_iz^i$, then it follows
\begin{align}\label{hp1}
[H_+(t),\phi(z)]=(\sum_{n\geq1,odd}t_nz^n)\phi(z),
\end{align}
and consequently, it turns out
\begin{align}\label{hp2}
e^{H_+(t)}\phi(z)e^{-H_+(t)}=e^{\xi(t,z)}\phi(z),
\end{align}
where $\xi(t,z)=\sum_{n\geq1,odd}t_nz^n$.
\end{proposition}
\begin{proof}
Before proving this result, firstly we should verify 
\begin{align*}
[H_n,\phi(z)]=z^n\phi(z).
\end{align*}
It comes true because
\begin{align*}
[H_n,\phi_k]&=\frac{1}{2}\sum_{i\in\mathbb{Z}}(-1)^{i+1}(\phi_i[\phi_{-i-n},\phi_k]_+-[\phi_i,\phi_k]_+\phi_{-n-i})\\&=\frac{1}{2}((-1)^{-n+1}\phi_{k-n}-(-1)\phi_{k-n})=\phi_{k-n},
\end{align*}
and 
\begin{align*}
[H_n,\phi(z)]=\sum_{k\in\mathbb{Z}}[H_n,\phi_k]z^k=\sum_{k\in\mathbb{Z}}\phi_{k-n}z^{k-n}z^n=z^n\phi(z).
\end{align*}
Therefore, equation (\ref{hp1}) can be obtained directly from the above computation and equation \eqref{hp2} can be verified by Lemma 3.35 in \cite{Hal}.
\end{proof}

\begin{remark}
It is noted that for the dual current operator
\begin{align*}
H_-(t)=\sum_{l\in 2\mathbb{Z}+1}t_{-l}H_{-l}=\frac{1}{2}\sum_{n\in\mathbb{Z},l\in2\mathbb{Z}+1}(-1)^{n+1}t_{-l}\phi_n\phi_{-n+l},
\end{align*}
it follows
\begin{align}
[H_-(t),\phi(z)]=(\sum_{l\in2\mathbb{Z}+1}t_{-l}z^{-l})\phi(z),\quad e^{H_-(t)}\phi(z)e^{-H_-(t)}=e^{\xi(t_-,z^{-1})}\phi(z),
\end{align}
which is of importance when the dual is taken into consideration in the derivations of integrable systems and Pfaffian point process.
\end{remark}

Now, we are about to demonstrate the Boson-Fermion correspondence of type $B_\infty$. Firstly, we consider the following spaces of even and odd elements respectively:
$$\mathcal{F}_B^0=\text{span}\{\phi_{n_1}\cdots\phi_{n_{2k}}|0\rangle\},\quad \mathcal{F}_B^1=\text{span}\{\phi_{n_1}\cdots\phi_{n_{2k+1}}|0\rangle\}$$ 
Then we have the decomposition $\mathcal{F}_B=\mathcal{F}_B^0\oplus \mathcal{F}_B^1$. 
Moreover, each $\mathcal{F}_B^i$ is isomorphic to $\bC[x_1,x_3,\cdots]$ and this isomorphism is called the Boson-Fermion correspondence of type $B_\infty$ \cite{you89}:

\begin{align*}
\sigma_B:\mathcal{F}_B\cong\bC[w,x_1,x_3,\cdots]/\sim\\
|U\rangle \rightarrow \sum\limits_{i=0}^1 w^i \langle i|e^{H_+(t)}|U\rangle,
\end{align*}
where $\sim$ is a relation $w^2\sim1$, i.e. we regard $w^2$ as $1$, and $|1\rangle=2^{\frac{1}{2}}\phi_0|0\rangle$. 

Note that $\sigma_B$ induces isomorphisms on each components. For $i=0,\;1$, one has
\begin{align*}
\sigma_B^i:\mathcal{F}_B^i\cong\bC[x_1,x_3,\cdots]\\
|U\rangle \rightarrow \langle i|e^{H_+(t)}|U\rangle.
\end{align*}
Under the isomorohism $\sigma_B$, for $n\geq 1$ odd, the Heisenberg action of $H_n$ can be realized on $\bC[x_1,x_3,\cdots]$ as $\displaystyle\frac{\partial}{\partial x_n}$, and $H_{-n}$ as $\displaystyle\frac{n}{2}x_n$.

Conversely, we can realize the fermionic action of $\phi_i$ on $\bC[w,x_1,x_3,\cdots]/<w^2-1>$. Let us introduce the vertex operator $X_B(z)=e^{\xi(t,z)}e^{-\xi(\tilde{\partial}_+,z^{-1})}$ and its dual $\bar{X}_B(z)=e^{-\xi(t_-,z^{-1})}e^{\xi(\tilde{\partial}_{-},z)}$ with notation
$
\tilde{\partial}_{\pm}=(2{\partial_{t_{\pm1}}},\frac{2}{3}\partial_{t_{\pm3}},\cdots),
$
then it follows
\begin{proposition}\label{prop:bosonfermion}
\begin{equation*}
\sigma_B(\phi(z)|U\rangle)=2^{-\frac{1}{2}}wX_B(z)\sigma_B(|U\rangle).
\end{equation*}
\end{proposition}

This proposition is stated in an equivalent form in \cite{date82,vanleur15,kac97}:
\begin{align*}
X_B(z)\langle 0|e^{H_+(t)}=2^{\frac{1}{2}}\langle 1| e^{H_+(t)}\phi(z),\quad X_B(z)\langle 1| e^{H_+(t)}=2^{\frac{1}{2}}\langle0|e^{H_+(t)}\phi(z).
\end{align*}
Since the proof is ignored in these articles, for readers' convenience, we give a brief proof in Appendix \ref{app:BKP}, which are also helpful in finding integrable hierarchies.

\section{Schur Q-functions and Pfaffian point process}\label{sec:Neu}

In this section, a Pfaffian point process is given in terms of neutral fermions. The basic facts of Okounkov's work are reviewed firstly. For a positive integer $n$, a partition $\lambda$ of $n$ is a set of positive integers $\lambda_1\geq\lambda_2\geq\cdots\geq\lambda_l>0$ with $n=\lambda_1+\cdots+\lambda_l$, and it is denoted by $\lambda=(\lambda_1,\cdots,\lambda_l)$. The number $|\lambda|=n$ is called the weight of $\lambda$, and $l(\lambda)=l$ is called the length of $\lambda$. We can also define the partition of 0, which is denoted by $\lambda=0$. In \cite{Okounkov2001Infinite}, Okounkov defined the Schur measure on partitions. For a partition $\lambda$, the measure of $\lambda$ is proportional to $s_{\lambda}(x)s_{\lambda}(y)$, where $s_{\lambda}$ is the Schur function. Under the Boson-Fermion correspondence, $s_{\lambda}$ can be written as the image of an element in the Fermionic Fock space, and then the correlation function is expressed as the vacuum expectation of certain fermions, which can be realized as a determinant point process by using Wick's theorem. It is shown that the determinants are tau functions of 2d-Toda hierarchy of integrable nonlinear partial differential equations of Ueno and Takasaki \cite{ueno84}.  A good review for this result can be found in \cite{borodin09}.

What we are going to do in this section is to apply Okounkov's discussion to strict partitions. A partition  $\lambda=(\lambda_1,\cdots,\lambda_l)$ is called a strict partition if $\lambda_1,\cdots,\lambda_l$ are distinct. The set of all strict partitions is denoted by DP. In \cite{tracy2004limit} the shifted Schur measure was defined on DP as follows:
\begin{align*}
\mathcal{M}(\lambda)=\frac{1}{Z}P_\lambda(x)Q_\lambda(y), \quad Z=\sum_{\lambda\in\text{DP}}P_\lambda(x)Q_\lambda(y)=\prod_{i,j}\frac{1+x_iy_j}{1-x_iy_j},
\end{align*}
where $P_\lambda(x)$ and $Q_\lambda(y)$ are Schur P-function and Schur Q-function respectively, which are defined later. For a finite subset $A\in\mathbb{Z}_+$, we can define the correlation function
\begin{align*}
\rho(A)=\mathcal{M}(\{\lambda, A\subset \sigma(\lambda)\}),
\end{align*}
where ${\sigma}(\lambda)=\{\lambda_i\}$. To compute the correlation function, we need to express $P_\lambda(x)$ and $Q_\lambda(y)$ as images of elements in the Fock space of neutral fermions under the Boson-Fermion correspondence of type $B_\infty$. Now let us give a definition of  $P_\lambda(x)$ and $Q_\lambda(y)$ in terms of the neutral fermions. 

Consider the expansion (cf. \cite{nimmo05,macdonald79})
\begin{align*}
e^{2\xi(t,z)}=\sum_{k\geq0}q_k(t)z^k,
\end{align*}
then from the equation \eqref{hp2}, one can conclude
\begin{align*}
e^{H_+(t)}\phi_ie^{-H_+(t)}=\sum_{k\geq0}q_k(\frac{1}{2}t)\phi_{i-k}.
\end{align*}
Then it is not difficult to compute that 
\begin{align}\label{eq:phiphi}
\begin{aligned}
\langle 0|e^{H_+(t)}\phi_i\phi_j|0\rangle&=\langle 0|e^{H_+(t)}\phi_ie^{-H_+(t)}e^{H_+(t)}\phi_je^{-H_+(t)}|0\rangle\\&=\frac{1}{2}q_i(\frac{1}{2}t)q_j(\frac{1}{2}t)+\sum_{k=1}^j(-1)^kq_{k+i}(\frac{1}{2}t)q_{j-k}(\frac{1}{2}t).
\end{aligned}
\end{align}
On the other hand, it is easy to compute the orthogonality condition
\begin{align*}
1=e^{2\xi(t,z)}e^{-2\xi(t,z)}=\sum_{i,j\geq0}q_i(t)q_{j}(-t)z^{i+j}.
\end{align*}
For $i+j=n>0$, one can obtain
\begin{align*}
\sum_{i=0}^n(-1)^iq_i(t)q_{n-i}(t)=0.
\end{align*}
This equation is trivial for $n$ odd and if $n=2m$, one can obtain
\begin{align*}
q_{m}(t)^2+2\sum_{k=1}^m(-1)^kq_{m+k}(t)q_{m-k}(t)=0,
\end{align*}
which helps to define
\begin{align}\label{eq:qab}
q_{a,b}(t)=q_a(t)q_b(t)+2\sum_{k=1}^b(-1)^kq_{a+k}(t)q_{b-k}(t)
\end{align}
with the property
\begin{align*}
q_{a,b}(t)=-q_{b,a}(t),\quad\quad q_{a,a}(t)=0.
\end{align*}
Moreover, from equations \eqref{eq:phiphi} and \eqref{eq:qab}, one knows
\begin{align*}
2\langle0| e^{H_+(t)}\phi_a\phi_b|0\rangle=q_{a,b}(\frac{1}{2}t).
\end{align*}
Therefore, for a strict partition with even length $\lambda=\{(\lambda_1,\cdots,\lambda_{2n})|\lambda_1>\cdots>\lambda_{2n}>0\}$, we can define a related Schur Q-function \cite{macdonald79,you89}
\begin{align*}
Q_\lambda(\frac{1}{2}t)=\Pf(q_{\lambda_i,\lambda_j}(\frac{1}{2}t))_{1\leq i,j\leq 2n}=\Pf(2\langle 0|e^{H_+(t)}\phi_i\phi_j|0\rangle)_{1\leq i,j\leq 2n}=2^{\frac{l(\lambda)}{2}}\langle0|e^{H_+(t)}\phi_{\lambda_1}\cdots\phi_{\lambda_{2n}}|0\rangle
\end{align*}
For a strict partition with odd length, the related Schur Q-function can be defined as
\begin{align*}
Q_\lambda(\frac{1}{2}t)&=\Pf\left(
\begin{array}{ccccc}
0&q_{\lambda_1,\lambda_2}(\frac{1}{2}t)&\cdots&q_{\lambda_{1},\lambda_{2n-1}}(\frac{1}{2}t)&q_{\lambda_1}(\frac{1}{2}t)\\
q_{\lambda_2,\lambda_1}(\frac{1}{2}t)&0&\cdots&q_{\lambda_2,\lambda_{2n-1}}(\frac{1}{2}t)&q_{\lambda_2}(\frac{1}{2}t)\\
\vdots&\vdots&&\vdots&\vdots\\
q_{\lambda_{2n-1},\lambda_1}(\frac{1}{2}t)&q_{\lambda_{2n-1},\lambda_2}(\frac{1}{2}t)&\cdots&0&q_{\lambda_{2n-1}}(\frac{1}{2}t)\\
-q_{\lambda_1}(\frac{1}{2}t)&-q_{\lambda_2}(\frac{1}{2}t)&\cdots&-q_{\lambda_{2n-1}}(\frac{1}{2}t)&0
\end{array}
\right)\\
&=\Pf(2\langle0|e^{H_+(t)}\phi_{\lambda_i}\phi_{\lambda j}|0\rangle)_{i,j=1,\cdots,2n-1,0}=2^{\frac{l(\lambda)+1}{2}}\langle 0|e^{H_+(t)}\phi_{\lambda_1}\cdots\phi_{\lambda_{2n-1}}\phi_0|0\rangle.
\end{align*}
With the Miwa transformation
\begin{align*}
t_n=\frac{2}{n}\sum_i x_i^n, \quad \text{$n$ odd},
\end{align*}
the following theorem is established \cite{you89}.
\begin{theorem}\label{schurq}
For a distinct partition $\lambda=(\lambda_1,\cdots,\lambda_l)\in \text{DP}$, we have
$$Q_\lambda(x)=2^{\frac{l(\lambda)}{2}}\sigma_B^{0}(\phi_{\lambda_1}\cdots\phi_{\lambda_l}|\alpha(\lambda)\rangle)=2^{\frac{l(\lambda)}{2}}\langle 0|e^{H_+(t)}\phi_{\lambda_1}\cdots\phi_{\lambda_l}|\alpha(\lambda)\rangle,$$
where
\begin{align*}
\alpha(\lambda)=\left\{
\begin{aligned}
&0\qquad\text{$l(\lambda)$ is even},\\
&1\qquad\text{$l(\lambda)$ is odd}.
\end{aligned}\right.
\end{align*}
\end{theorem}
\begin{remark}\label{rmk:ShurQ}
One can also define the Schur Q-function by $t_-$ part, which corresponds to the other Miwa variables
\begin{align*}
t_{-n}=\frac{2}{n}\sum_{i}y_i^n,\quad \text{n odd.}
\end{align*} 
For the distinct partition $\lambda\in\text{DP}$, one has
$$Q_\lambda(y)=(-1)^{|\lambda|}2^{\frac{l(\lambda)}{2}}\langle \alpha(\lambda)|\phi_{-\lambda_l}\cdots\phi_{-\lambda_1}e^{H_-(t)}|0\rangle.$$
\end{remark}

Another symmetric function, the Schur P-function $P_\lambda(x)$, can be defined as
\begin{align*}
P_\lambda(x)=2^{-l(\lambda)} Q_\lambda(x)
\end{align*}
such that $\langle P_\lambda, Q_\mu\rangle=\delta_{\lambda,\mu}$. It is straightforward to write $P_{\lambda}$ in the terms of neutral fermions by Theorem \ref{schurq} and Remark \ref{rmk:ShurQ}.

Now we are ready to demonstrate the following theorem.
\begin{theorem}\label{thm:Paf}
The correlation function can be expressed in terms of neutral fermions as
\begin{align*}
\rho(A)=\frac{1}{Z}\sum_{A\subset \sigma(\lambda)}P_\lambda(x)Q_\lambda(y)=\frac{1}{Z}\langle 0|e^{H_+(t)}(\prod_{a\in A}(-1)^a\phi_a\phi_{-a})e^{H_-(t)}|0\rangle.
\end{align*}
Moreover, it can be expressed as a Pfaffian point process
\begin{align}\label{ppp}
\rho(A)=\Pf(K(a_i,a_j))_{a_{ i},a_{ j}\in\pm A},\quad K(a,b)=\langle 0|e^{H_+(t)}e^{-H_-(t)}\phi_a\phi_be^{H_-(t)}e^{-H_+(t)}|0\rangle.
\end{align}
\end{theorem}

\begin{proof}
By using theorem \ref{schurq}, one can basically obtain
\begin{align*}
&\langle 0|e^{H_+(t)}=\sum_{\lambda\in \text{DP}} (-1)^{|\lambda|}2^{\frac{l(\lambda)}{2}}P_\lambda(x)\langle\alpha(\lambda) |\phi_{-\lambda_l}\cdots\phi_{-\lambda_1},\\
&e^{H_-(t)}|0\rangle=\sum_{\lambda\in\text{DP}}2^{-\frac{l(\lambda)}{2}}Q_\lambda(y)\phi_{\lambda_1}\cdots\phi_{\lambda_l}|\alpha(\lambda)\rangle.
\end{align*}
Moreover, for two partitions $\lambda=(\lambda_1,\cdots,\lambda_l)$ and $\mu=(\mu_1,\cdots,\mu_k)$, one has
$$\langle\alpha(\lambda)|\phi_{-\lambda_l}\cdots\phi_{-\lambda_1}\phi_{\mu_1}\cdots\phi_{\mu_{k}}|\alpha(\mu)\rangle=(-1)^{|\lambda|}\delta_{\lambda,\mu},$$
then it is true that
\begin{align*}
\sum_{A\subset\sigma(\lambda)}P_\lambda(x)Q_\lambda(y)=\langle 0|e^{H_+(t)}(\prod_{a\in A}(-1)^a\phi_a\phi_{-a})e^{H_-(t)}|0\rangle.
\end{align*}
In addition, since
\begin{align*}
\log Z&=\sum_{i,j}(\log(1+x_iy_j)-\log(1-x_iy_j))=\sum_{i,j}\sum_{k>0}\frac{1}{k}(1-(-1)^k)x_i^ky_j^k\\
&=\sum_{i,j}\sum_{n>0,odd} \frac{2}{n}x_i^ny_j^n=\sum_{n>0,odd}\frac{n}{2}t_nt_{-n}
\end{align*}
and by the use of \eqref{z}, we can get $e^{H_+(t)}e^{-H_-(t)}=Ze^{H_-(t)}e^{-H_+(t)}$. Finally, if we denote $G=e^{H_+(t)}e^{-H_-(t)}$ and $\Phi_i=G\phi_iG^{-1}$, it follows
\begin{align*}
\rho(A)&=\frac{1}{Z}\langle 0|e^{H_+(t)}(\prod_{a\in A}(-1)^a\phi_a\phi_{-a})e^{H_-(t)}|0\rangle\\
&=\frac{1}{Z}\langle 0|e^{H_+(t)}G^{-1}(\prod_{a\in A}(-1)^a\Phi_a\Phi_{-a})Ge^{H_-(t)}|0\rangle\\
&=\langle 0|e^{H_-(t)}(\prod_{a\in A}(-1)^a\Phi_a\Phi_{-a})e^{H_+(t)}|0\rangle=\langle 0|\prod_{a\in A}(-1)^a\Phi_a\Phi_{-a}|0\rangle.
\end{align*}
Thus, by Wick's theorem, we know it is a Pfaffian with expression \eqref{ppp}.
\end{proof}

\begin{remark}
Theorem \ref{thm:Paf} has been derived in \cite{matsumoto2005correlation} from a different point of view. See \cite{matsumoto2005correlation,vuletic2007shifted} for more details on the properties of the correlation functions of the shifted Schur measure and the shifted Schur process. In our paper, $\rho(A)$ is constructed using netural fermions, and thus it is related natually to tau functions of BKP hierarchies, which are discussed in Appendix \ref{app:BKP}.
\end{remark}

\section{matrix integrals solution of BKP hierarchy and Pfaffian point process}\label{sec:Mat}
There have been several examples between the matrix integrals solution of integrable hierarchy and determinantal point process and the fact lies in the $\tau$-functions (matrix integrals solution) of those integrable hierarchies can be viewed as the normalization constant of those point processes. In \cite{adler95}, the matrix integrals solution of Toda (KP) hierarchy  and determinantal point process in the configuration space in $\mathbb{R}^n$ have been discussed with details. The determinantal point process in the configuration space in $\mathbb{C}^n$ and its connection with matrix integrals solution of Topelitz lattice are shown in \cite{vanmoerbeke2011Random}. 

In this part, we mainly discuss the matrix integrals solution of BKP hierarchy, which can help us to induce a novel Pfaffian point process. For this purpose, firstly, it has been demonstrated that the partition function of Bures ensemble with suitable time flows can be regarded as the $\tau$-function of BKP hierarchy \cite{hu17,OST,orlov16}.
\begin{proposition}
When time flows are introduced, the partition function of Bures ensemble
\begin{align}\label{eq:tau}
\tau_n=\frac{1}{n!}\int_{\mathbb{R}_+^n}\prod_{1\leq i<j\leq n}\frac{(x_i-x_j)^2}{x_i+x_j}\prod_{i=1}^n\omega(x_i;t)dx_i,\quad\omega(x;t)=\omega(x)\exp(\sum_{k\geq1,odd}t_kx^k)
\end{align}
can be viewed as the $\tau$-function of BKP hierarchy for some formal weight $\omega(x)$. Moreover, if we denote the moments $$\omega_{i,j}=\iint_{\mathbb{R}_+^2}\frac{x-y}{x+y}x^iy^j\omega(x;t)\omega(y;t)dxdy,\quad\omega_i=\int_{\mathbb{R}_+}x^i\omega(x;t)dx$$ then this partition function can be written in terms of Pfaffian as
\begin{align*}
\tau_n=\left\{
 \begin{array}{ll}
 \Pf\left(\begin{array}{c}
\omega_{i,j}
\end{array}
\right)_{i,j=1}^{2m} & n=2m,\\
\ \\
\Pf\left(
\begin{array}{ccc}
0&\vline&
\begin{array}{cc}
\omega_i
\end{array}
\\
\hline
\begin{array}{c}
-\omega_j
\end{array}
&\vline&\omega_{i,j}
\end{array}
\right)_{i,j=1}^{2m+1} &n=2m+1.
 \end{array} 
 \right. 
\end{align*}
\end{proposition}

This proposition is consistent with the fact that the irreducible highest weight representation is splitted into two families $\mathcal{F}_B^0\oplus\mathcal{F}_B^1$ and each of them has individual form. Remarkablely, this $\tau$-function of BKP hierarchy can be constructed from the vacuum expectation form \eqref{ve} if we choose the group like element $G$ appropriately. Now we would like to demonstrate how to relate this $\tau$-function of BKP equation to a Pfaffian point process proposed in \cite{rains00}.

\begin{proposition}\label{pfaffianpointprocess}
The matrix integrals solution $\tau_{2m}$ of BKP hierarchy in \eqref{eq:tau} can induce a Pfaffian point process. That is, let $(X,d\mu)$ be a measure space, $\{\phi_i(x)=x^i, \, i=0,\cdots,2m-1\}$ be functions from $X$ to $\mathbb{C}$ and $\epsilon(x,y)=\frac{x-y}{x+y}$ be skew symmetric function from $X\times X\to \mathbb{C}$, then moments
\begin{align*}
M_{i,j}=\int_{x,y\in X}\phi_i(x)\phi_j(y)\epsilon(x,y)d\mu(x)d\mu(y)
\end{align*}
form an invertible antisymmetric matrix $M$\footnote{It is noted that in this case, $X$ should be a configuration space in $\mathbb{R}_+^{2m}$ to ensure the existence of the moments $M_{i,j}$.}. For a finite subset $S=\{x_1,\cdots,x_l\}\subset X$ with $l\leq 2m$, we can define a correlation function
\begin{align}\label{eq:correlation}
R(S)=\frac{1}{(2m-l)!\pf(M)}\int_{x_{l+1},\cdots,x_{2m}\in X}\prod_{1\leq i<j\leq 2m}\frac{(x_i-x_j)^2}{x_i+x_j}\prod_{l+1\leq j\leq 2m} d\mu(x_j),
\end{align}
which is related to a Pfaffian point process with kernel {\fontsize{9pt}{11pt}
\begin{align*}
K(x,y)=\left(\begin{array}{cc}
\sum_{0\leq i,j\leq 2m-1}\phi_i(x)M^{-T}_{i,j}\phi_j(y)&\sum_{0\leq i,j\leq 2m-1}\phi_i(x)M^{-T}_{i,j}(\epsilon\cdot\phi_j)(y)\\
\sum_{0\leq i,j\leq 2m-1}(\epsilon\cdot\phi_i)(x)M^{-T}_{i,j}\phi_j(y)&-\epsilon(x,y)+\sum_{0\leq i,j\leq 2m-1}(\epsilon\cdot\phi_i)(x)M^{-T}_{i,j}(\epsilon\cdot\phi_j)(y)\end{array}
\right),
\end{align*}}
where $(\epsilon\cdot f)(x)=\int_{y\in X}\epsilon(x,y)f(y)dx$.
\end{proposition}

The fact that correlation function \eqref{eq:correlation} is related to the Pfaffian point process with a skew symmetric kernel $K(x,y)$ has been proved in \cite{rains00}. Here we just take specific $\{\phi_i(x)\}_{i=0}^{2m-1}$ as monomials and $\epsilon(x,y)$ as specific skew symmetric inner product kernel so that $\pf(M)$ is indeed the $\tau$-function of BKP hierarchy if proper time flows are introduced. Moreover, the invertibility of skew symmetric matrix $M$ is based on Pfaffian Schur identity and de Bruijn formula \cite{debruijn55,forrester16}. Then a natural question is: whether the $\tau_{2m+1}$ of BKP equation can induce a such Pfaffian point process? To do so, a notation about Pfaffian of odd order need to be firstly clarified.

As usual, a Pfaffian is only defined on skew-symmetric matrices of even order; however, it takes some advantages to consider odd orders. In \cite{debruijn55}, Pfaffian can be defined in any $n\times n$ skew-symmetric  matrix $A$ as
\begin{align}\label{eq:pf}
\Pf(A)=\frac{1}{2^m m!}\sum_{j_1=1}^n\cdots\sum_{j_n=1}^n\sigma\left(\begin{array}{ccc}j_1&\cdots& j_n\\
1&\cdots&n\end{array}\right)a_{j_1j_2}\cdots a_{j_{2m-1}j_{2m}},\quad m=[\frac{1}{2}n]
\end{align}
with $\sigma$ the perturbation of set $\{j_1,\cdots,j_n\}$ to $\{1,\cdots,n\}$. It can be found that the even case is just the same as the original definition of Pfaffian, but the odd ones involve something new. Indeed, let $K$ be a skew-symmetric $n\times n$ matrix with $n$ odd and $K^+$ arises from $K$ by adding an $(n+1)$-th column consisting of $n$ elements $1$, an $(n+1)$-th row consisting of $n$ elements $-1$ and $k_{n+1,n+1}^+=0$. Then one can show $\Pf(K)$ defined by \eqref{eq:pf} is the same with $\Pf(K^+)$. Later on, we denote $\Pf(K^+)$ as the Pfaffian of $K$ modified by the above process.

\begin{proposition}
Let $(X,d\mu)$ be a measure space and $\{\phi_i,\,i=0,\cdots,2m\}$ be $2m+1$ functions from $X$ to $\mathbb{C}$. Assume that $\epsilon(x,y)$ is a skew-symmetric function from $X\times X\to \mathbb{C}$ and related moments are defined as
\begin{align*}
&M_{i,j}=\int_{x,y\in X}\phi_i(x)\epsilon(x,y)\phi_j(y)d\mu(x)d\mu(y),\quad 0\leq i,j\leq 2m,\\
&M_{i,2m+1}=-M_{2m+1,i}=\int_{x\in X}\phi_i(x) d\mu(x),\quad\,\,\,\, 0\leq i\leq 2m.
\end{align*}
Assume the skew-symmetric moment matrix $M_{(2m+2)\times(2m+2)}$ is well-defined and invertible, then for a finite subset $S=\{x_1,\cdots,x_l\}\subset X$ with $l\leq 2m+1$, we can define a correlation function 
\begin{align*}
R(S;\phi,\epsilon)=\frac{1}{(2m+1-l)!\Pf(M^+)}\int_{x_{l+1},\cdots,x_{2m+1}\in X}\det(\phi_i(x_j))\Pf(\epsilon^+(x_i,x_j))\prod_{l+1\leq j\leq 2m+1}d\mu(x_j).
\end{align*}
For $|S|>2m+1$, $R(S;\phi,\epsilon)=0$ and for $|S|\leq 2m+1$, one has $R(S;\phi,\epsilon)=\Pf(K^+(S))$, where $K^+$ is a skew-symmetric matrix kernel
\begin{align*}
&K^+(x,y)=\left(\begin{array}{cc|cc}
&\hspace{-22ex}{\raisebox{-1.5ex}[0pt]{$K(x,y)$}}&\sum_{i=0}^{2m}\phi_i(x)M_{i,2m+1}^{-T}&0\\ 
&&\sum_{i=0}^{2m}(\epsilon\cdot\phi_i)(x)M_{i,2m+1}^{-T}&-1\vspace{2pt} \\ \hline  
\raisebox{-2pt}{$\sum_{i=0}^{2m}M_{2m+1,i}^{-T}\phi_i(y)$}&\raisebox{-2pt}{$\sum_{i=0}^{2m}M_{2m+1,i}^{-T}(\epsilon\cdot\phi_i)(y)$}&\raisebox{-2pt}{$0$}&\raisebox{-2pt}{$0$}\\
0&1&0&0
\end{array}
\right)
\end{align*}
with
\begin{align*}
K(x,y)=\left(\begin{array}{cc}
\sum_{0\leq i,j\leq 2m}\phi_i(x)M^{-T}_{i,j}\phi_j(y)&\sum_{0\leq i,j\leq 2m}\phi_i(x)M^{-T}_{i,j}(\epsilon\cdot\phi_j)(y)\\
\sum_{0\leq i,j\leq 2m}(\epsilon\cdot\phi_i)(x)M^{-T}_{i,j}\phi_j(y)&-\epsilon(x,y)+\sum_{0\leq i,j\leq 2m}
(\epsilon\cdot\phi_i)(x)M^{-T}(\epsilon\cdot\phi_j)(y)
\end{array}
\right)
\end{align*}
and $\epsilon$ is an operator defined by $(\epsilon\cdot f)(x)=\int_{y\in X}\epsilon(x,y)f(y)d\mu(y)$.
\end{proposition}
\begin{proof}
It is of importance to state the invertibility of moments matrix $M$ is equal to the non-singularity of $\det(\phi_i(x_j))\pf(\epsilon^+(x_i,x_j))$ due to the de Bruijn integral formula. Moreover, the proof includes three different cases, that is, we need to prove the cases $|S|=2m+1$, $|S|>2m+1$ and $|S|<2m+1$. Firstly, we are going to prove the case $|S|=2m+1$. Note that $\Pf(K(S))$ is a Pfaffian of $(4m+4)\times(4m+4)$ skew symmetric matrix and by basic row and column transformations, one can obtain
\begin{align*}
\Pf(K(S))=(-1)^{m+1}\left(\begin{array}{cc}
RM^{-T}R^{T}&RM^{-T}(\epsilon\cdot R)^{T}\\
(\epsilon\cdot R)M^{-T}R^T&-\epsilon^++(\epsilon\cdot R)M^{-T}(\epsilon\cdot R)^T
\end{array}\right)
\end{align*}
where $\epsilon^+$ is a matrix extended by the matrix $\epsilon$ and $R$ is a $(2m+2)\times(2m+2)$ matrix
\begin{align*}
R=\left(\begin{array}{cccc}
\phi_1(x_1)&\cdots&\phi_{2m+1}(x_1)&0\\
\vdots&&\vdots&\vdots\\
\phi_1(x_{2m+1})&\cdots&\phi_{2m+1}(x_{2m+1})&0\\
0&\cdots&0&1
\end{array}\right).
\end{align*}
Therefore, by expressing $(\epsilon\cdot\phi_j)(x)$ on the set $S$ as a linear combination of the functions $\phi_j(x)$, we can find
\begin{align*}
\Pf(K(S))&=(-1)^{m+1}\Pf\left(
\left(\begin{array}{cc}
I&\\
L&I\end{array}\right)\left(\begin{array}{cc}
RM^{-T}R^T&\\&-\epsilon^+\end{array}\right)\left(
\begin{array}{cc}
I&L^T\\
&I\end{array}\right)
\right).
\end{align*}
with $L$ a coefficient matrix of transformation from $(\epsilon\cdot R)$ to $R$. Noting the formula $\Pf(ABA^{T})=\det(A)\Pf(B)$ holds for all antisymmetric matrix $B$, the above equation can be rewritten as 
\begin{align*}
\Pf(K(S))=\det(R)\Pf(\epsilon^+)\Pf(M^{-T})=\frac{1}{\Pf(M)}\det(\phi_i(x_j))\Pf(\epsilon^+(x_i,x_j)).
\end{align*} 
 For the case $|S|>2m+1$, $R(S;\phi,\epsilon)=0$ is obvious since the matrix $R$ is singular if the elements of $S$ is more than $2m+1$. In addition, for the case $|S|<2m+1$, we just need to prove the following equality due to an induction
 \begin{align*}
\int_{x_l\in X}\Pf(K(S\cup\{x_l\}))d\mu(x_l)=(2m-l+2)\Pf(K(S)).
\end{align*}
This equality can be proved by using the method proposed in \cite{rains00} and we omit here. Thus we say the kernel $K^+$ here can induce a Pfaffian point process constituting of odd points.
\end{proof}

If one chooses $\epsilon(x,y)=\frac{x-y}{x+y}$ and $\phi_i(x)=x^i$, this Pfaffian point process coincides with the one induced by Bures ensemble with odd points \cite{forrester16}. Here we just show a general case for the Pfaffian point process with a skew-symmetric matrix kernel, which admits odd points and is different with the model proposed in \cite{rains00}. Moreover, the Pfaffian point process of even and odd points can be formulated by the unified one, if we use the de Bruijn's notation about Pfaffian.

\begin{proposition}
Let $(X,d\mu)$ be a measure space and $\{\phi_i,i=1,\cdots,n\}$ be $n$ functions from $X$ to $\mathbb{C}$. Assume that $\epsilon(x,y)$ is a skew symmetric function from $X\times X$ to $\mathbb{C}$ and the related partition function
\begin{align*}
\tau_n=\int_{x_1,\cdots,x_n\in X}\det(\phi_i(x_j))_{i,j=1}^n\Pf(\epsilon(x_i,x_j))_{i,j=1}^n\prod_{1\leq j\leq n}d\mu(x_j)
\end{align*}
is well-defined and invertible, 
then for a finite subset $S=\{x_1,\cdots,x_l\}\subset X$ with $l\leq n$, we can define an $l$-point correlation function
\begin{align*}
R(S;\phi,\epsilon)=\frac{1}{(n-l)!\tau_n}\int_{x_{l+1},\cdots,x_n\in X}\det(\phi_i(x_j))_{i,j=1}^n\Pf(\epsilon(x_i,x_j))_{i,j=1}^n\prod_{l+1\leq j\leq n}d\mu(x_j).
\end{align*}
Moreover, this correlation function is related to a skew symmetric matrix kernel $K$ for even $n$ and $K^+$ for odd $n$ with $R(S;\phi,\epsilon)=\Pf(K(S))$ (or $\Pf(K^+(S))$ respectively) if $|S|\leq n$ and $R(S;\phi,\epsilon)=0$ if $|S|>n$.
\end{proposition}

\section*{Acknowledgements}
The author Z. Wang is supported by the Fundamental Research Funds for the Central Universities, Grant No. 00-800015ND and the author S. Li would like to thank Professor K. Takasaki for telling us the (modified) negative flow of BKP hierarchy has already appeared in the appendix of \cite{JM}, and thank  Dr. M. Wheeler for helpful comments. 

\begin{appendix}
\section{Applications to Tau function and BKP hierarchy}\label{app:BKP}
In this appendix, we review some facts about the neutral fermions and tau functions of BKP hierarchies. 

We start with the proof of Proposition \ref{prop:bosonfermion}. Notice that it suffices to prove the following lemma.
\begin{lemma}\label{lem:fermionAction}
There hold
\begin{equation*}
\langle 1|e^{-H_+[z^{-1}]}=2^{\frac12}\lz\phi(z)
\quad
\text{and}\quad
\langle 0|e^{-H_+[z^{-1}]}=2^{\frac12}\langle 1|\phi(z)
\end{equation*}
with $[z]=(2z,\frac{2z^3}{3},\frac{2z^5}{5},\cdots)$.
\end{lemma}

\begin{proof}
Let us prove the first identity, which is equivalent to the following identity
\begin{equation}\label{identity}
\lz\phi(z)\phi(z_1)\cdots\phi(z_{2s-1})\rz=\lz\phi_0e^{-H_+[z^{-1}]}\phi(z_1)\cdots\phi(z_{2s-1})\rz.
\end{equation}
By using equation \eqref{eqn:VacExp}, one has
\begin{equation*}
LHS=\displaystyle\frac{1}{2^s}\prod\limits_{i=1}^{2s-1}\displaystyle\frac{1-z_i/z}{1+z_i/z}\prod\limits_{j<j'}\displaystyle\frac{1-z_{j'}/z_j}{1+z_{j'}/z_j}.
\end{equation*}
Moreover, the RHS of \eqref{identity} is the coefficient of $k^0$ in $\lz\phi(k)e^{-H_+[z^{-1}]}\phi(z_1)\cdots\phi(z_{2s_1})|0\rangle$. It can be computed as 
\begin{equation*}
\begin{split}
\lz\phi(k)e^{-H_+[z^{-1}]}\phi(z_1)\cdots\phi(z_{2s_1})\rz
=&\prod\limits_{i=1}^{2s-1}e^{-\xi([z^{-1}],z_i)}\lz\phi(k)\phi(z_1)\cdots\phi(z_{2s-1})e^{-H_+[z^{-1}]}\rz\\
=&\prod\limits_{i=1}^{2s-1}e^{-\xi([z^{-1}],z_i)}\lz\phi(k)\phi(z_1)\cdots\phi(z_{2s-1})\rz\\
=&\prod\limits_{i=1}^{2s-1}\displaystyle\frac{1-z_i/z}{1+z_i/z}\cdot\displaystyle\frac{1}{2^s}\prod\limits_{i=1}^{2s-1}\displaystyle\frac{1-z_i/k}{1+z_i/k}\prod\limits_{j<j'}\displaystyle\frac{1-z_{j'}/z_j}{1+z_{j'}/z_j},
\end{split}
\end{equation*}
from which one knows the coefficient of $k^0$ is $\displaystyle\frac{1}{2^s}\prod\limits_{i=1}^{2s-1}\displaystyle\frac{1-z_i/z}{1+z_i/z}\prod\limits_{j<j'}\displaystyle\frac{1-z_{j'}/z_j}{1+z_{j'}/z_j}$ indeed.
\end{proof}

\begin{remark}
It also holds true for the dual vertex operator that
\begin{align*}
\bar{X}_B(z)\cdot e^{H_-(t)}|0\rangle=2^{\frac{1}{2}}\phi(z)e^{H_-(t)}|1\rangle\quad\text{and}\quad \bar{X}_B(z)\cdot e^{H_-(t)}|1\rangle=2^{\frac{1}{2}}\phi(z)e^{H_-(t)}|0\rangle.
\end{align*}
\end{remark} 

We see soon that Lemma \ref{lem:fermionAction} is used in calculations of tau functions of BKP hierarchies.

Firstly we introduce the basic bilinear condition of BKP hierarchy and give some discussions about the tau function of this hierarchy, which was called the small BKP hierarchy in \cite{OST} 
\begin{align*}
\sum_{n\in\mathbb{Z}}(-1)^n\phi_nG\otimes\phi_{-n}G=\sum_{n\in\mathbb{Z}}(-1)^{n}G\phi_n\otimes G\phi_{-n},
\end{align*}
or an equivalent matrix form as
\begin{align}\label{bbc}
\sum_{n\in\mathbb{Z}}(-1)^n\langle U|G\phi_n|V\rangle\langle U'|G\phi_{-n}|V'\rangle=
\sum_{n\in\mathbb{Z}}(-1)^n\langle U|\phi_nG|V\rangle\langle U'|\phi_{-n}G|V'\rangle.
\end{align}
Noting that different hierarchies are obtained by choosing different $\langle U|$, $\langle U'|$, $|V\rangle$ and $|V'\rangle$, now we constrain ourselves to the simplest case, that is, the BKP hierarchy.
Taking $\langle U|=\langle 0|\phi_0e^{H_+(t)}$, $\langle U'|=\langle 0|\phi_0e^{H_+(t')}$ and $|V\rangle=|V'\rangle=|0\rangle$, and assume the tau function has the form \cite{JM}
\begin{align*}
\tau(t)=\langle 0|e^{H_+(t)}G|0\rangle=2\langle 0|\phi_0 e^{H_+(t)}G\phi_0|0\rangle,
\end{align*}
then one has
\begin{align*}
\frac{1}{4}\tau(t)\tau(t')&=\sum_{n\in\mathbb{Z}}(-1)^n\langle 0|\phi_0e^{H_+(t)}\phi_nG|0\rangle\langle 0|\phi_0e^{H_+(t')}\phi_{-n}G|0\rangle\\
&=\text{Res}(z^{-1}\langle 0|\phi_0e^{H_+(t)}\phi(z)G|0\rangle\langle 0|\phi_0e^{H_+(t')}\phi(-z)G|0\rangle )\\
&=\frac{1}{4}\text{Res}(z^{-1}e^{\xi(t-t',z)}\langle 0|e^{H_+(t-[z^{-1}])}G|0\rangle\langle 0|e^{H_+(t'+[z^{-1}])}G|0\rangle),
\end{align*}
where Lemma \ref{lem:fermionAction} is used to derive the last identity.
 
It equals the following bilinear form 
\begin{align*}
\oint_{C_\infty}\frac{dz}{2\pi iz}e^{\xi(t-t',z)}\tau(t-[z^{-1}])\tau(t'+[z^{-1}])=\tau(t)\tau(t'),
\end{align*}
and this is the bilinear equation of BKP hierarchy indeed.

In what follows, a negative flow of BKP hierarchy is given when $t_{-1}$, $t_{-3}$, $\cdots$ are introduced in tau functions \cite{OST}. Taking $\langle U|=\langle 0|\phi_0e^{H_+(t)}$, $\langle U'|=\langle 0|\phi_0e^{H_+(t')}$, $|V\rangle=e^{H_-(t)}|0\rangle$ and $|V'\rangle=e^{H_-(t')}|0\rangle$, and define the tau function
\begin{align*}
\tau(t_+,t_-)=\langle 0|e^{H_+(t)}Ge^{H_-(t)}|0\rangle=2\langle 0|\phi_0 e^{H_+(t)}Ge^{H_-(t)}\phi_0|0\rangle,
\end{align*}
then the negative flow of BKP hierarchy can be obtained, for which we state as the following proposition.
\begin{proposition}
There exists a negative BKP hierarchy
\begin{align*}
&\oint_{C_\infty}e^{\xi(t_+-t'_+,z)}\tau(t_+-[z^{-1}],t_-)\tau(t'_++[z^{-1}],t'_-)\frac{dz}{2\pi iz}\\
&=\oint_{C_0}e^{\xi(t'_--t_-,z^{-1})}\tau(t_+,t_-+[z])\tau(t'_+,t'_--[z])\frac{dz}{2\pi iz}.
\end{align*}
\end{proposition}
The proof of this proposition is similar with that of BKP hierarchy and we omit it here. By applying the transformation
\begin{align*}
t_+\to t_+-a,\, t'_+\to t_++a,\, t_-\to t_--b,\, t'_-\to t_-+b,
\end{align*}
one can obtain the following bilinear equations
\begin{align*}
\sum_{k\geq0}p_k(-2a)p_k(\tilde{D}_+)e^{\sum_{l\geq 1,odd}a_lD_l+b_lD_{-l}}\tau\cdot\tau=\sum_{m\geq 0}p_m(2b)p_m(-\tilde{D}_-)e^{\sum_{l\geq 1,odd}a_lD_l+b_lD_{-l}}\tau\cdot\tau,
\end{align*}
where $p_k$ is a homogeneous symmetric function and a detailed introduction is given in the next section. Here one should notice that if the index of the negative flow vanishes, this hierarchy goes back to the original BKP hierarchy. The first member of this hierarchy is
\begin{align*}
D_{-1}(D_3-D_1^3)\tau\cdot\tau=0,
\end{align*}
which has appeared in \cite{JM} as an equation derived from $D'_\infty$ and the negative BKP hierarchy here seems to be embedded as a subsystem of the $D'_\infty$ hierarchy. The negative flow plays an important role since it involves the dual of the original current operator, which is closely related to the construction of the correlation function.

Besides the negative flow of BKP hierarchy, the modified BKP (mBKP) hierarchy is also important \cite{date82}. It exhibits two different kinds of tau functions based on the decomposition of the Fock space $\mathcal{F}_B=\mathcal{F}_B^0\oplus \mathcal{F}_B^1$ and it provides the original idea to consider a Pfaffian point process of even and odd points in section \ref{sec:Mat}. Moreover, the mBKP hierarchy holds for the following proposition.

\begin{proposition}
There exists an mBKP hierarchy
\begin{align}\label{mbkp}
\oint_{C_\infty} e^{\xi(t-t',z)}\tau_{n+1}(t-[z^{-1}])\tau_{n}(t'+[z^{-1}])\frac{dz}{2\pi iz}=2\tau_n(t)\tau_{n+1}(t')-\tau_{n+1}(t)\tau_n(t').
\end{align}
\end{proposition}

\begin{proof}
According to the basic bilinear condition \eqref{bbc}, if one takes
\begin{align*}
\langle U|=\langle0|e^{H_+(t)}\phi(z),\quad \langle U'|=\langle 0|\phi_0e^{H_+(t')},\quad |V\rangle=|V'\rangle=|0\rangle,
\end{align*}
then it follows
\begin{align*}
\text{Res}(\langle 0|e^{H_+(t)}\phi(z)\phi(z')G|0\rangle\langle0|\phi_0e^{H_+(t')}\phi(-z')G|0\rangle)=\langle 0|e^{H_+(t)}\phi(z)G\phi_0|0\rangle\langle0|\phi_0e^{H_+(t')}G\phi_0|0\rangle.
\end{align*}
Moreover, if we set tau functions as
\begin{align*}
\tau_n=\langle 0|e^{H_+(t)}G|0\rangle=2\langle 0|\phi_0e^{H_+(t)}G\phi_0|0\rangle,\quad \tau_{n+1}=2\langle 0|e^{H_+(t)}\phi(z)G\phi_0|0\rangle=2\langle 0|\phi_0e^{H_+(t)}\phi(z)G|0\rangle,
\end{align*}
the above equation can be rewritten as
\begin{align*}
\oint_{z'=\infty}\langle 0|e^{H_+(t)}\phi(z)\phi(z')G|0\rangle\langle 0|\phi_0e^{H_+(t')}\phi(-z')G|0\rangle\frac{dz'}{2\pi iz'}=\frac{1}{4}\tau_{n+1}(t)\tau_{n}(t').
\end{align*}
Notice that $\phi(z)\phi(z')=\delta(-\frac{z}{z'})-\phi(z')\phi(z)$ and use the Boson-Fermion correspondence, one can easily obtain the equation \eqref{mbkp}, and thus complete the proof.
\end{proof}

\begin{remark}
This hierarchy firstly appeared in \cite{date82} without proof and it is noted that the two different tau functions $\tau_n$ and $\tau_{n+1}$ are linked to each other by the vertex operator. In fact, according to the Boson-Fermion correspondence of type $B_\infty$,
one can see that the tau functions are of different expressions when Pfaffian is involved, and the solutions are split into two families. Examples can be found in \cite{hirota04,hu17}.
\end{remark}
Indeed, the above mBKP hierarchy (\ref{mbkp}) can be equivalently expressed in a bilinear form as
\begin{align*}
\left[\sum_{l\geq1}h_l(-2a)h_l(2\tilde{D}_+)e^{\sum_{j=1}a_jD_j}+4\sinh(\sum_{j=1}a_jD_j)\right]\tau_n\cdot\tau_{n+1}=0.
\end{align*}
By considering that $a$ is of one component, the above equation is
\begin{align*}
\sum_{m\geq1}a^m\left[\sum_{k+l=m,l\geq1,k\geq0}(\frac{(-2)^l}{l!k!}p_l(2\tilde{D}_+)D_1^k)+4\frac{D_1^m}{m!}\right]\tau_n\cdot\tau_{n+1}=0,\quad\text{for $m$ odd}.
\end{align*}
The first two nontrivial equations are the cases of $m=3$ and $m=5$:
\begin{align*}
(D_1^3-D_3)\tau_n\cdot\tau_{n+1}=0,\quad (6D_5-5D_3D_1^2-D_1^5)\tau_n\cdot\tau_{n+1}=0,
\end{align*}
which are the modified BKP equations (or the B\"acklund transformation of BKP equation).

Here we would like to mention that the negative flow of the modified BKP equation is also included in the $D'_\infty$ hierarchy \cite{JM}. Consider the tau function with time flows $t_+$ and $t_-$, one can define\footnote{It should be noted that the different choices of group-like elements $G$ can help us to construct different kinds of solutions of these integrable hierarchies and this is the idea why the matrix integrals solution can induce some Pfaffian point process.}
\begin{align}\label{ve}
\begin{aligned}
\tau_n&=\langle 0|e^{H_+(t)}Ge^{H_-(t)}|0\rangle=2\langle 0|\phi_0e^{H_+(t)}Ge^{H_-(t)}\phi_0|0\rangle,\\ \tau_{n+1}&=2\langle 0|e^{H_+(t)}\phi(z)Ge^{H_-(t)}\phi_0|0\rangle=2\langle 0|\phi_0e^{H_+(t)}\phi(z)Ge^{H_-(t)}|0\rangle.
\end{aligned}
\end{align}
Based on the equation \eqref{bbc} and choosing
\begin{align*}
\langle U|=\langle0|e^{H_+(t)}\phi(z),\quad \langle U'|=\langle 0|\phi_0e^{H_+(t')},\quad |V\rangle=e^{H_-(t)}|0\rangle,\quad
|V'\rangle=e^{H_-(t')}|0\rangle,
\end{align*}
 one can obtain
\begin{align*}
&\oint_{C_0}e^{\xi(t'_--t_-,z^{-1})}\tau_{n+1}(t_+,t_-+[z])\tau_n(t'_+,t'_--[z])\frac{dz}{2\pi iz}=\\
&2\tau_n(t_+,t_-)\tau_{n+1}(t'_+,t'_-)-\oint_{C_\infty}e^{\xi(t_+-t'_+,z)}\tau_{n+1}(t_+-[z^{-1}],t_-)\tau_n(t'_++[z^{-1}],t'_-)\frac{dz}{2\pi iz},
\end{align*}
and its first member is $D_1D_{-1}\tau_n\cdot\tau_{n+1}=0$, which has been listed in the appendix of \cite{JM}.
\end{appendix}

\small
\bibliographystyle{abbrv}

\end{document}